\documentclass[aps,pra,twocolumn,amsmath,showpacs,letterpaper,floatfix]{revtex4}

\newcommand{\bra}[1]{\left\langle #1\right|}
\newcommand{\ket}[1]{\left| #1\right\rangle}
\newcommand{\braket}[2]{\left\langle #1|#2\right\rangle}

\usepackage{amssymb,amsmath,amsthm,amsfonts,relsize}
\usepackage{graphicx}
\usepackage{epsfig}
\usepackage{subfigure}
\usepackage{amsmath}
\usepackage{bm}			
\usepackage{color}
\newtheorem*{theorem}{Theorem}

\begin{document}

\title{Verification of quantum discord}

\author{Saleh Rahimi-Keshari$^{1}$, Carlton M. Caves$^{2,3}$, and Timothy C. Ralph$^{1}$}
\affiliation{$^{1}$Centre for Quantum Computation and Communication Technology,
School of Mathematics and Physics,
University of Queensland, St Lucia, Queensland 4072, Australia \\
$^{2}$Center for Quantum Information and Control, University of New Mexico,
MSC07-4220, Albuquerque, New Mexico 87131-0001, USA\\
$^{3}$Centre for Engineered Quantum Systems, School of Mathematics and Physics,
University of Queensland, St Lucia, Queensland 4072, Australia}

\begin{abstract}
We introduce a measurement-based method for verifying quantum discord of any bipartite quantum system. We show that by performing an informationally complete POVM (IC-POVM) on one subsystem and checking the commutativity of the conditional states of the other subsystem, quantum discord from the second subsystem to the first can be verified. This is an improvement upon previous methods, which enables us to efficiently apply our method to continuous-variable systems, as IC-POVMs are readily available from homodyne or heterodyne measurements.  We show that quantum discord for Gaussian states can be verified by checking whether the peaks of the conditional Wigner functions corresponding to two different outcomes of heterodyne measurement coincide at the same point in the phase space. Using this method, we also prove that the only Gaussian states with zero discord are product states; hence, Gaussian states with Gaussian discord have nonzero quantum discord.
\end{abstract}

\pacs{42.50.Dv, 03.65.Ta, 03.65.Ud}

\maketitle

\section{Introduction}
\label{Intro}

Quantum correlations play a central role as a resource in quantum information processing and quantum communication tasks.  Traditionally, entanglement was thought of as the unique form of quantum correlation and the reason why quantum computers can outperform classical computers.
Yet there are tasks that are believed to be exponentially hard classically, which can be done efficiently using quantum computational models with little or no entanglement~\cite{KL98,DFCaves05}.  Quantum discord was introduced as a more general measure of quantum correlation for bipartite systems, with no classical analogue~\cite{Zurek01}, and it was suggested as a resource for certain quantum computation models~\cite{DSCaves08}, quantum state merging~\cite{madhok11,cavalcanti}, and for encoding information onto a quantum state~\cite{gu}.  Discord has been generalized to continuous-variable systems to study quantum correlations in Gaussian states~\cite{Datta10,Paris10} and certain nonGaussian states~\cite{TMAdessoK12}.

Recently, schemes have been proposed to test for nonvanishing quantum discord of discrete-variable quantum states~\cite{terno2010,rahimi2010,BC10,DVBrukner10,modi2011,ZYCO11,sarandy2011,long2011,GAdesso12,serra2012}, and some of these have been implemented in nuclear-magnetic-resonance systems~\cite{Serra11,Laflamme11} and in an optical system~\cite{walborn2012}.
Of particular practical interest, however, is a general method for detecting nonvanishing discord in the joint state of both discrete and continuous-variable systems.

In this paper we introduce a measurement-based method for verifying quantum discord of any bipartite quantum state, without requiring any prior knowledge of the joint state.
We consider the post-measurement states of one subsystem, $B$, conditioned to all the outcomes of an informationally complete POVM (IC-POVM) performed on the other subsystem, $A$.  We show that if the post-measurement states of $B$ commute with one another, then the quantum discord from $B$ to $A$ is zero. Conversely, if they do not commute, the quantum discord from $B$ to $A$ is necessarily nonzero. A POVM is informationally complete if its outcome probabilities are sufficient to determine uniquely the quantum state, i.e., to perform quantum state tomography~\cite{Prug77,Busch91}.  In other words, a bipartite state has zero discord from $B$ to $A$ if tomography on $A$ leaves the eigenstates of the density operator of $B$ unchanged.

Our method for verifying quantum discord is an improvement on the existing method~\cite{modi2011}, as it only requires measurement of one IC-POVM on $A$.
Hence, this method can be readily applied to continuous-variable systems where an IC-POVM is available from either heterodyne or homodyne measurements. We discuss in Sec.~\ref{Method} that the commutativity of the conditional states of $B$ can be efficiently tested by checking the commutation relations between one nondegenerate conditional state and all other states.

Quantum discord is defined as the difference between two classically equivalent measures for mutual information~\cite{Zurek01}.  According to Bayes's rule for classically correlated probability distributions, the quantities $I(A:B)=H(A)+H(B)-H(A,B)$, $J(A|B)=H(A)-H(A|B)$, and $J(B|A)=H(B)-H(B|A)$, where $H$ denotes the Shannon entropy and $H(A|B)=H(A,B)-H(B)$ is the conditional entropy, are all equal; they are called the classical mutual information.  For a bipartite quantum system, the quantum mutual information is defined, in analogy to $I(A:B)$, by $I(\rho_{AB})=S(\rho_A)+S(\rho_B)-S(\rho_{AB})$, where $S(\rho)=-\text{Tr}[\rho\log(\rho)]$ is the von Neumann entropy.  A measurement-based, quantum version of the conditional entropy is $S_{\{\Pi_j\}}(A|B)=\sum_{j}p_{j}S(\rho_{A|j})$, where $p_j=\text{Tr}[\rho_{AB}\Pi_{j}]$, $\rho_{A|j}=\text{Tr}_B[\rho_{AB}\Pi_j]/p_j$, and the set $\{\Pi_j\}$, with $\sum_j\Pi_j=\mathbb{I}$, makes up a POVM measurement on subsystem $B$.  This conditional entropy depends on the choice of measurement; hence, the quantum analogue of $J(A|B)$ is defined by minimizing over all possible measurements: $J^{\leftarrow}(\rho_{AB})=S(\rho_{A})-\text{inf}_{\{\Pi_j\}}S_{\{\Pi_j\}}(A|B)$.

The quantum discord from $B$ to $A$ is then defined as the difference between these two ways of defining a quantum mutual information:
\begin{align}
D^{\leftarrow}(\rho_{AB})&=I(\rho_{AB})-J^{\leftarrow}(\rho_{AB}) \nonumber \\
&=S(\rho_{B})-S(\rho_{AB})+\text{inf}_{\{\Pi_j\}}S_{\{\Pi_j\}}(A|B)\;.
\end{align}
The quantum discord is zero if and only if the quantum state can be expressed in the form
\begin{equation}
\rho_{AB}=\sum_{j} p_{j} \rho_{j}\otimes\ket{j}\bra{j\,}\;,
\label{form}
\end{equation}
where $\{\ket{j}\}$ are orthogonal states and $0\leq p_{j}\leq 1$~\cite{Zurek01,MLang11}.  For a quantum state with this form, local measurements on $B$ in the basis $\{\ket{j\,}\bra{j\,}\}$ leave the system unperturbed, and all the state information can be extracted without joint measurements.
Notice that the state~\eqref{form} is diagonal in a conditional product basis pointing from $B$ to $A$, i.e., an orthogonal basis of the form $\{|f_{jk}\rangle\otimes|e_j\rangle\}$, where the states $|f_{jk}\rangle$ are the eigenstates of $\rho_j$. Hence, for a given quantum state, quantum discord can be directly verified by diagonalizing the joint density operator~\cite{MLang11}.

Although there is no general method for minimizing over all possible measurements in order to calculate the conditional entropy of a state, this can sometimes be done when there are restrictions to certain classes of states and POVMs.  Thus the Gaussian quantum discord is defined as the quantum discord for two-mode Gaussian states where the evaluation of the conditional entropy is restricted to generalized Gaussian measurements~\cite{Datta10,Paris10}.

This paper is structured as follows.  In the next section, we introduce the method for verifying quantum discord.  In Sec.~\ref{CVsystems}, we discuss the application of the method to continuous-variable systems.  In Sec.~\ref{Gaussian}, we propose a technique for verifying quantum discord of Gaussian states. Based on that, we show that for Gaussian states, only product states have zero discord; hence states with nonzero Gaussian discord have nonvanishing quantum discord.

\section{Method}
\label{Method}

A previously proposed measurement-based method~\cite{modi2011} for verifying quantum discord is based on testing whether a quantum state can be expressed in the form~\eqref{form} of states with zero discord. We improve this previous proposal by showing that the quantum discord can be verified with only one IC-POVM.

\begin{theorem}
For a bipartite system $\rho_{AB}$, the necessary and sufficient condition for having zero discord
from $B$ to $A$, $D^{\leftarrow }(\rho_{AB})=0$, is that the states of subsystem $B$, $\rho_{B|k}=\text{Tr}_{A}[M_{k}\rho_{AB}]/\text{Tr}_{AB}[M_{k}\rho_{AB}]
=\text{Tr}_{A}[M_{k}\rho_{AB}]/p_k$, conditioned to the outcomes $k$ of an IC-POVM on $A$ (POVM elements $\{M_{k}\}$), commute with one another, i.e.,
\begin{equation}
\left[\rho_{B|k},\rho_{B|k'}\right]=0\,,\quad\text{for any}\ k\ \text{and}\ k'\,.
\label{condition}
\end{equation}
\end{theorem}
\begin{proof}
For a state having the zero-discord form~\eqref{form},
\begin{equation}
\rho_{B|k}=
\frac{\sum_jp_j\text{Tr}_A[M_k\rho_j]|j\,\rangle\langle j\,|}
{\sum_jp_j\text{Tr}_A[M_k\rho_j]}\;,
\end{equation}
immediately demonstrating that the states $\rho_{B|k}$ are all diagonal in the basis $\{|j\rangle\}$ and thus commute.

For the converse, we assume the condition~\eqref{condition}.  For such a set of commuting conditional states $\{\rho_{B|k}\}$, there exists an orthonormal basis, $\{\ket{j}\}$, that diagonalizes all the conditional states, $\rho_{B|k}=\sum_{j}\lambda_{kj}\ket{j\,}\bra{j\,}$.  To say that the POVM elements $\{M_k\}$ make up an IC-POVM is to say that they span the space of operators and thus there exist operators $\{N_k\}$ such that $\rho=\sum_k N_k\text{Tr}[M_k\rho]$ for any density operator $\rho$.  Applying this identity to the joint state gives
\begin{equation}
\rho_{AB}=\sum_k N_k\text{Tr}_A[M_k\rho_{AB}]=\sum_{k,j} \lambda_{kj}p_kN_k\otimes\ket{j}\bra{j}\;.
\end{equation}
Hence, $\rho_{AB}$ has the form~\eqref{form}, with $\rho_j=\sum_k\lambda{kj}p_kN_k$.
\end{proof}

Physically, what the proof says is that for a state of zero discord from $B$ to $A$, the measurement of the POVM $\{M_k\otimes\ket{j\,}\bra{j\,}\}$ extracts all information about the state $\rho_{AB}$.  From the perspective of the original definition of discord~\cite{Zurek01}, one imagines extracting this information by first measuring $B$ in the basis $\{|j\rangle\}$ and then measuring an IC-POVM on $A$.  Our criterion for zero discord works from the opposite perspective by reversing the order of the measurements on $A$ and $B$.

In order to test whether an unknown quantum state has nonzero discord experimentally, based on this theorem, one needs to measure an IC-POVM on subsystem $A$ and determine, by state tomography for each outcome, the corresponding states of the subsystem~$B$.
This procedure continues until one of the commutation relations between conditional states of subsystem $B$ is nonzero.
If subsystem $A$ has a $d$-dimensional Hilbert space, one can always find an IC-POVM that has $d^2$ POVM elements~\cite{Prug77,Caves02}.  Hence, there are $d^2$ conditional states of subsystem~$B$ and $d^2(d^2-1)/2$ commutation relations between all pairwise states.
However, as the conditional states are Hermitian operators, the most efficient way to check commutativity is to calculate the commutation relations between one of the states with no degeneracy and all other states. In this case, there are at most $d^2-1$ commutation relations to be checked.
Also, if some prior knowledge about the state in question is available, as is often the case in practice, quantum discord can be tested by considering only a few IC-POVM elements.
 Consider, for example, the maximally entangled state, $\ket{\psi}=\sum_{j=1}^d\ket{j\,}\ket{j\,}/\sqrt d$.  Any two rank-one outcomes, $\ket{n}$ and $\ket{\eta}$, on one of the subsystems, provided $0<|\braket{n}{\eta}|<1$ (these could be outcomes from two distinct, nonorthogonal projective measurements), yield conditional states of the other subsystem that do not commute.  Also, as we show below, for Gaussian states only two different heterodyne outcomes are sufficient to verify quantum discord.

\section{Continuous-variable systems}
\label{CVsystems}

An interesting feature of this method is that it can be readily applied to continuous-variable systems, as complete sets of IC-POVMs are available from heterodyne or homodyne detection. Two sets of measurements are required, one on each of the subsystems.  In general, one needs to do state tomography to construct the quasiprobability distributions of subsystem $B$ for all the states conditioned to outcomes of the measurement performed on subsystem~$A$.  Then the commutativity of the states $\rho_{B|k}$, which are represented in terms of quasiprobabilities, must be checked in order to verify discord. This can be efficiently done by finding one nondegenerate state and calculating the commutation relations between that state and all other states using an appropriate relation in terms of the reconstructed quasiprobabilities.  For instance, if the Wigner functions  $W_{B|k}(\alpha)$ of conditioned states of subsystem $B$ are available, the commutation relations between corresponding density operators can be calculated by using the Moyal Bracket~\cite{Moyal}:
\begin{align}
W_{kk'}(\alpha)=\frac{1}{2\pi}\int\text{d}^2\!\beta\,&\text{d}^2\!\beta'\,
W_{B|k}\bigl(\alpha+\textstyle{\frac{1}{2}}\beta\bigr)
W_{B|k'}\bigl(\alpha+\textstyle{\frac{1}{2}}\beta'\bigr) \nonumber \\
&\times \sin\!\left(i\frac{\beta\beta^{\prime*}-\beta'\beta^*}{2}\right)\;.
\end{align}
Here $W_{kk'}(\alpha)$ is the Wigner-like function for the operator $-i[\rho_{B|k},\rho_{B|k'}]$. If the states commute with each other then $W_{kk'}(\alpha)=0$ for all $\alpha$. Alternatively, the commutation relations can be calculated using characteristic functions,
\begin{align}
\mathlarger{\chi}_{kk'}(\xi)
=\frac{2}{\pi}\int\text{d}^2\zeta\,&
\mathlarger{\chi}_{B|k}\bigl({\textstyle\frac{1}{2}}\xi+\zeta\bigr)
\mathlarger{\chi}_{B|k'}\bigl({\textstyle\frac{1}{2}}\xi-\zeta\bigr)\nonumber \\ &\times\sin\!\left(i\frac{\xi\zeta^*-\xi^*\zeta}{2}\right)\;,
\end{align}
or in terms of any other quasiprobability distributions~\cite{Agarwal-Wolf}.

For states with zero discord, the eigenstates of the conditional density operator of $B$ do not change while $A$ is being fully determined from measurements of an IC-POVM. For continuous-variable systems, defined on an infinite-dimensional Hilbert space, the IC-POVM will have an infinite number of outcomes.  In practice, only a finite number of measurement outcomes can be explored.  For instance, in homodyne detection only a finite number of phases are considered, and the phase space is subdivided into a finite number of bins.  This introduces errors in the state estimation and uncertainties for the reconstructed quasiprobabilities of the conditional states, which propagate to the distributions representing the operators $-i[\rho_{B|k},\rho_{B|k'}]$. If one of these commutator distributions takes on a nonzero value at some point, which is larger than its associated uncertainty, then quantum discord is necessarily nonzero; otherwise, it is not clear whether the discord is nonzero.
However, by having some prior knowledge about the state, such as being Gaussian, the error can be estimated, and it can be made arbitrarily small using a sufficiently large number of measurements.

\section{Gaussian states}
\label{Gaussian}

A special class of continuous-variable states consists of the Gaussian states, i.e., those states whose Wigner function is a Gaussian function.  Such states are uniquely characterized by the means and covariance matrix of their quadrature components, $x$ and $p$.  For two systems, with modal annihilation operators $\hat a=x_1+ip_1$ and $\hat b=x_2+ip_2$, we define quadrature vectors for each system, $\mathbf{x}_1=(x_1,p_1)$ and $\mathbf{x}_2=(x_2,p_2)$, and we define an overall quadrature vector $\mathbf{x}=(\mathbf{x}_1,\mathbf{x_2})=(x_1,p_1,x_2,p_2)$.

The means of the quadrature components can be set to zero by locally displacing the two systems.  Then the state is specified by its covariance matrix~\cite{Adesso-Illuminati}
\begin{equation}
\bm{\sigma}=\langle\mathbf{x}^T\mathbf{x}\rangle=
\begin{pmatrix}
\mathbf{A} & \mathbf{C} \\
\mathbf{C}^{T} & \mathbf{B}
\end{pmatrix}
\;.
\label{sigmagen}
\end{equation}
Using local unitary operations that preserve the Gaussian form of the states, the covariance matrix of a bipartite Gaussian state can be brought to a standard form in which
$\mathbf{A}=\text{diag}(a,a)$, $\mathbf{B}=\text{diag}(b,b)$, and $\mathbf{C}=\text{diag}(c,d)$, where $a\ge0$ and $b\ge0$.  This can be accomplishing by first applying local unitary rotations that diagonalize $\mathbf{A}$ and $\mathbf{B}$, then using local squeezing operations to transform these diagonal blocks to $\mathbf{A}=\text{diag}(a,a)$ and $\mathbf{B}=\text{diag}(b,b)$, and finally applying further local unitary rotations to diagonalize $\mathbf{C}$. Notice that positivity of the density operator imposes the uncertainty-principle constraint~\cite{RSimon00},
\begin{align}
&\bm{\sigma}+\frac{i}{4}\bm{\Omega}\ge0\;,\quad
\Omega=\begin{pmatrix}
\mathbf{J} & \mathbf{0} \\
\mathbf{0} & \mathbf{J}
\end{pmatrix}\;,\quad
&\mathbf{J}=\begin{pmatrix}
0&1\\-1&0
\end{pmatrix}\;.
\end{align}
For a covariance matrix in standard form, this implies that $a^2\ge1/16$, $b^2\ge1/16$, $ab\ge c^2$, and $ab\ge d^2$, plus cubic and quartic constraints on $a$, $b$, $c$, and $d$.

It has been shown that Gaussian discord for a two-mode Gaussian state is zero if and only if $\mathbf{C}=0$~\cite{Datta10,Paris10}.  Here we show this condition is also necessary and sufficient for having zero discord.

A zero-mean Gaussian state with the standard form of the covariance matrix has characteristic function $\text{\large${\chi}$}(\mathbf{k})=\langle e^{i\mathbf{k}\mathbf{x}^T}\rangle=e^{-\mathbf{x}\bm{\sigma}\mathbf{x}^T/2}$
and Wigner function
\begin{align}
&W(\mathbf{x}_1,\mathbf{x}_2)=W(\mathbf{x})\nonumber\\
&=\frac{1}{4\pi^2\sqrt{\det\bm{\sigma}}}
\,\exp\!\left(-\frac{\mathbf{x}\bm{\sigma}^{-1}\mathbf{x}}{2}\right)\nonumber\\
&=\frac{1}{4\pi^2\sqrt{(ab-c^2)(ab-d^2)}}\nonumber\\
&\times\exp\!\!\left(-\frac{bx_1^2+ax_2^2-2cx_1x_2}{2(ab-c^2)}-\frac{bp_1^2+ap_2^2-2dp_1p_2}{2(ab-d^2)}\right)\;.
\end{align}

Suppose Alice makes a heterodyne measurement on subsystem~$A$; i.e., she uses the IC-POVM whose POVM elements are the coherent states $\ket{\beta}\!\bra{\beta}$.  Let $\beta=x'_1+ip'_1$ specify the outcomes of here measurement.  Then the state $\rho_{B|\mathbf{x}'_1}$, conditioned on these outcomes, has Wigner function
\begin{align}
W_{B|\mathbf{x'_1}}(\mathbf{x}_2)=\frac{1}{N}\int\text{d}^2\mathbf{x}_1\,
W(\mathbf{x}_1,\mathbf{x}_2)W_{\mathbf{x}'_1}(\mathbf{x}_1)\;,
\end{align}
where $W_{\mathbf{x}'_1}(\mathbf{x}_1)=2\exp[-2(x_1-x'_1)^2-2(p_1-p'_1)^2]/\pi$ is the Wigner function of coherent state $\ket{\beta}=\ket{x'_1+ip'_1}$ and $N$ is a normalization factor. Integration yields
\begin{align}
W_{B|\mathbf{x}'_1}(\mathbf{x}_2)=&\frac{1}{N'} \exp\!\left(-\frac{1}{2}\,x_2^2f(a,b,c)+x_2x'_1g(a,b,c)\right)\nonumber\\
&\times \exp\!\left(-\frac{1}{2}\,p_2^2f(a,b,d)+p_2p'_1g(a,b,d)\right)\;,
\label{WB}
\end{align}
where
\begin{align}
f(a,b,z)&=\frac{1}{ab-z^2}\left(a-\frac{z^2}{b+4(ab-z^2)}\right)\;, \nonumber \\
g(a,b,z)&=\frac{4z}{b+4(ab-z^2)}\;, \nonumber
\end{align}
and $N'$ is a normalization factor.

The peak of system~$B$'s conditional Wigner function~\eqref{WB}, located at
\begin{equation}
\gamma=\frac{g(a,b,c)}{f(a,b,c)}x'_1+i\frac{g(a,b,d)}{f(a,b,d)}p'_1\;,
\label{peak}
\end{equation}
depends on the outcomes of the measurement on $A$, $x'_1$ and $p'_1$, unless $c=0$ and $d=0$. Consequently, the eigenvectors of the conditional state $\rho_{B|\mathbf{x}'_1}$, which are generally displaced, squeezed number states, $\{D(\gamma)S(\zeta)\ket{n}\}$, displaced to the Wigner-function peak $\gamma$, change based on the outcomes of the heterodyne measurement performed on subsystem~$A$. This indicates nonzero discord, since the eigenvectors do not commute.  Therefore, without explicitly calculating any commutation relations, we can see that the bipartite Gaussian state has nonvanishing discord unless $c=0$ and $d=0$. Transforming back from the standard form to the general convariance matrix~\eqref{sigmagen}, one can say that a bipartite Gaussian state has zero discord if and only if $\mathbf{C}=0$, i.e., if and only if the state is a product state. This also implies that states with Gaussian discord ($\mathbf{C}\neq0$) have nonzero quantum discord.

These results show that quantum discord of Gaussian states can be verified using only two different heterodyne outcomes on one subsystem and finding (by tomography) the points in the phase space at which the corresponding conditional Wigner functions attain their maximum values. If those points do not coincide, the quantum discord is nonzero, since having different peaks guarantees that the corresponding eigenstates, which are displaced, squeezed number states, do not commute.
This argument can also be applied to nonGaussian states: if there are two conditional Wigner functions with the same shape, but located at different points in phase space, they correspond to states $\rho$ and $D(\nu)\rho D^{\dagger}(\nu)$, which have two different sets of eigenvectors $\{\ket{\psi_i}\}$ and $\{D(\nu)\ket{\psi_i}\}$, which is sufficient evidence that the quantum discord is nonzero.  Note that discord exists even if only one of $c$ and $d$ is nonzero, so to uncover discord of Gaussian states with only two heterodyne outcomes, one should choose the outcomes to be different for both quadratures.

\section{Conclusion}
We have introduced a method for verifying quantum discord of any bipartite quantum system.  The method is based on the fact that all the information in states with zero discord from subsystem $B$ to subsystem $A$ can be fully extracted by measurements that are diagonal in a single basis of $B$. In order to verify discord, one needs to perform an IC-POVM on subsystem~$A$ and check whether the conditional density operators of the subsystem~$B$ commute, i.e., whether they share the same eigenstates.

It is worth mentioning that, in practice, one would check commutativity of the conditional states as an IC-POVM is being performed on $A$. This can be efficiently done by finding a nondegenerate conditional state and calculate the commutation relations between that state and other states. In this case, the maximum number of commutation relations to be checked scales linearly with the number of IC-POVM elements.  Once one of these commutators is found to be nonzero, that confirms nonzero discord.  This method can be simply applied on continuous-variable systems by using homodyne or heterodyne detection and calculating the commutation relations in terms of quasiprobability distributions.  We have shown that a bipartite Gaussian state has nonzero quantum discord if and only if it is not a product state, which is the same as the condition for having nonzero Gaussian discord. Moreover, we show that with only two heterodyne outcomes and without calculating any commutation relations, quantum discord of Gaussian states can be verified.

\vspace{18pt}

\acknowledgments

We thank M.~Gu and N.~Walk for useful discussions.  This research was conducted by the Australian Research Council Centre of Excellence for Quantum Computation and Communication
Technology (Project number CE11000102) and was partially supported by U.S.\ National Science Foundation Grant Nos.~PHY-1005540, PHY-0903953, and PHY-1212445.


\end{document}